\documentclass[11pt]{article}
\usepackage[english]{babel}
\usepackage[T1]{fontenc}
\usepackage[utf8]{inputenc}
\usepackage{amsmath}
\usepackage{amssymb}
\usepackage{amsfonts}
\usepackage{amsthm}
\usepackage{bm}
\usepackage{mathtools}
\usepackage{upgreek}

\newcommand{\bc}{\bm{c}}

\newcommand{\eps}{\varepsilon}
\newcommand{\C}{\mathbb{C}}
\newcommand{\N}{\mathbb{N}}

\newcommand{\R}{\mathbb{R}}
\renewcommand{\S}{\mathbb{S}}

\newcommand{\boA}{\mathcal{A}}

\newcommand{\boC}{\mathcal{C}}

\newcommand{\boE}{\mathcal{E}}

\newcommand{\boH}{\mathcal{H}}
\newcommand{\boI}{\mathcal{I}}

\newcommand{\boN}{\mathcal{N}}

\newcommand{\boV}{\mathcal{V}}

\newcommand{\ga}{\mathfrak{a}}
\newcommand{\gb}{\mathfrak{b}}
\newcommand{\gc}{\mathfrak{c}}

\newcommand{\gs}{\mathfrak{s}}

\newcommand{\gv}{\mathfrak{v}}

\newcommand{\Pos}{{\rm Pos}}

\newtheorem*{claim*}{Claim}
\newtheorem{claim}{Claim}

\newtheorem{lem}{Lemma}
\newtheorem{prop}{Proposition}

\newtheorem{thm}{Theorem}

\theoremstyle{definition}

\newtheorem*{merci}{Acknowledgments}
\newtheorem{remark}{Remark}

\theoremstyle{remark}

\begin{document}
\title{On the asymptotic stability in the energy space for multi-solitons of the Landau-Lifshitz equation }
\author{\renewcommand{\thefootnote}{\arabic{footnote}} Yakine Bahri
\footnotemark[1]}
\footnotetext[1]{Centre de Math\'ematiques Laurent Schwartz, \'Ecole polytechnique, 91128 Palaiseau Cedex,
France. E-mail: {\tt yakine.bahri@polytechnique.edu}}
\maketitle
\begin{abstract}
We establish the asymptotic stability of multi-solitons for the one-dimensional Landau-Lifshitz equation with an easy-plane anisotropy. The solitons have non-zero speed, are ordered according to their speeds and have sufficiently separated initial positions.
We provide the asymptotic stability around solitons and between solitons. More precisely, we show that for an initial datum close to a sum of $N$ dark solitons, the corresponding solution converges weakly to one of the solitons in the sum, when it is translated to the centre of this soliton, and converges weakly to zero  when it is translated between solitons.  

\end{abstract}
\numberwithin{cor}{section}
\numberwithin{remark}{section}
\numberwithin{lem}{section}
\numberwithin{prop}{section}
\numberwithin{thm}{section}
\section{Introduction}
We consider the one-dimensional Landau-Lifshitz equation
\renewcommand{\theequation}{LL}
\begin{equation}
\label{LL}
\partial_t m + m \times( \partial_{xx} m + \lambda m_3 e_3) = 0,
\end{equation}
for a map $m = (m_1, m_2, m_3) : \R \times \R \to \S^2$, where $e_3 = (0, 0, 1)$ and $\lambda \in \R$. This equation, which was introduced by Landau and Lifshitz in \cite{LandLif1}, describes the dynamics of magnetization in a one-dimensional ferromagnetic material, for example in  CsNiF$_3$ or TMNC (see e.g. \cite{KosIvKo1, HubeSch0} and the references therein). $\lambda$ is the anisotropy parameter of the material. The case $\lambda > 0$ gives account for an easy-axis anisotropy and the case $\lambda < 0$ of an easy-plane anisotropy. The equation reduces to the one-dimensional Schrödinger map equation in the isotropic case $\lambda = 0$. This equation has been intensively studied (see e.g. \cite{BeIoKeT1,GuoDing0,JerrSme1}). In this paper, we are interested in the easy-plane anisotropy case ($\lambda < 0)$. Scaling the map $m$, if necessary, we can assume from now on $\lambda = - 1$.

The Hamiltonian for the Landau-Lifshitz equation, the so-called Landau-Lifshitz energy, is given by
$$E(m) := \frac{1}{2} \int_\R \big( |\partial_x m|^2 + m_3^2 \big).$$
In this paper, we study the solutions $m$ to \eqref{LL} with finite Landau-Lifshitz energy, i.e. which belong to the energy space
$$\boE(\R) := \big\{ \upsilon : \R \to \S^2, \ {\rm s.t.} \ \upsilon' \in L^2(\R) \ {\rm and} \ \upsilon_3 \in L^2(\R) \big\}.$$
A soliton with speed $c$ is a travelling-wave solution of \eqref{LL} which has the form
$$m(x, t) := u(x - c t).$$
Its profile $u$ is solution to the ordinary differential equation
\renewcommand{\theequation}{TWE}
\begin{equation}
\label{Solc}
u''+ |u'|^2 u + u_3^2 u - u_3 e_3 + c u \times u' = 0.
\end{equation}
The solutions of this equation are explicit.  If $|c| < 1$, there exist non-constant solutions $u_c$ to \eqref{Solc}, which are given by the formulae
\begin{equation*}
[u_c]_1(x) = \frac{c}{\cosh \big( (1 - c^2)^\frac{1}{2} x \big)}, \quad [u_c]_2(x) = \tanh \big( (1 - c^2)^\frac{1}{2} x \big), \quad [u_c]_3(x) = \frac{(1 - c^2)^\frac{1}{2}}{\cosh \big( (1 - c^2)^\frac{1}{2} x \big)},
\end{equation*}
up to the invariances of the problem, i.e. translations, rotations around the axis $x_3$ and the orthogonal symmetry with respect to the plane $x_3 = 0$ (see \cite{deLaire4} for more details). Else, when $|c| \geq 1$, the only solutions with finite Landau-Lifshitz energy are the constant vectors in $\S ^1 \times \{0\}$.

In dimension one the equation is completely integrable using the inverse scattering method (see e.g.~\cite{FaddTak0}). This method allows to justify the existence of multi-solitons for \eqref{LL} and to compute their expression (see \cite{Mikhail1, Rodin1}). Multi-solitons, which can be considered as a nonlinear superposition of single solitons, are exact solutions to \eqref{LL}. Our main goal is to prove the asymptotic stability of  multi-solitons (see Theorem~\ref{thm:stabasympt_m1} below).

Martel, Merle and Tsai proved the asymptotic stability for multi-solitons of the subcritical gKdV equations in \cite{MarMeTs1}. Martel and Merle stated this result for one soliton of the generalized KdV equation in \cite{MartMer2} and then they refined the results for multi-solitons in \cite{MartMer5}. This method was successfully adapted by Bethuel, Gravejat and Smets to prove the asymptotic stability for a dark soliton of the Gross-Pitaevskii equation in \cite{BetGrSm2} and then in \cite{Bahri} to show the same result for the Landau-Lifshitz equation. Cuccagna and Jenkins proved similar results for the Gross-Pitaevskii equation in \cite{Cuccagn3} using the inverse scattering method. Perelman established the asymptotic stability of multi-solitons for the nonlinear Schrödinger equation in \cite{Perelman}. 

In the next subsections, we first introduce the hydrodynamical framework in which we provide all the analysis and we provide our main result. 

\subsection{The hydrodynamical framework}
We denote by $\check{m}$ the map defined by $\check{m} := m_1 + i m_2.$ We have
$$|\check{m}(x)| = (1 - m_3^2(x))^\frac{1}{2} \to 1,$$
as $x \to \pm \infty$, using the fact that $m_3$ belongs to $H^1(\R),$ and the Sobolev embedding theorem. 
This allows us, as in the case of the Gross-Pitaevskii equation (see e.g. \cite{BetGrSa2}), to consider the hydrodynamical framework for the Landau-Lifshitz equation. In terms of the maps $\check{m}$ and $m_3$, this equation may be written as
$$\Bigg\{ \begin{array}{ll}
i \partial_t \check{m} - m_3 \partial_{xx} \check{m} + \check{m} \partial_{xx} m_3 - \check{m} m_3 = 0,\\[5pt]
\partial_t m_3 + \partial_x \big\langle i \check{m}, \partial_x \check{m} \big\rangle_\C = 0.
\end{array}$$
When the map $\check{m}$ does not vanish, one can write it as $\check{m} = (1 - m_3^2)^{1/2} \exp i \varphi$. The hydrodynamical variables $v := m_3$ and $w := \partial_x \varphi$ verify the following system
\renewcommand{\theequation}{HLL}
\begin{equation}
\label{HLL}
\left\{ \begin{array}{ll}
\partial_t v = \partial_x \big( (v^2 - 1) w \big),\\[5pt]
\displaystyle \partial_t w = \partial_x \Big( \frac{\partial_{xx} v}{1 - v^2} + v \frac{(\partial_x v)^2}{(1 - v^2)^2} + v \big( w^2 - 1) \Big).
\end{array} \right.
\end{equation}
This system is similar to the hydrodynamical Gross-Pitaevskii equation (see e.g. \cite{BetGrSm2}).\footnote{The hydrodynamical terminology originates in the fact that the hydrodynamical Gross-Pitaevskii equation is similar to the Euler equation for an irrotational fluid (see e.g. \cite{BetGrSm1}).}
The Cauchy problem in the space $X(\R):=H^1(\R)\times L^2(\R)$ for this system was solved by de Laire and Gravejat in \cite{DeLGr}, where local well-posedness is established.

In this framework, the Landau-Lifshitz energy is expressed as
\renewcommand{\theequation}{\arabic{equation}}
\numberwithin{equation}{section}
\setcounter{equation}{0}
\begin{equation}
\label{eq:hydro-E}
E(\gv) := \int_\R e(\gv) := \frac{1}{2} \int_\R \Big( \frac{(v')^2}{1 - v^2} + \big( 1 - v^2 \big) w^2 + v^2 \Big),
\end{equation}
where $\gv := (v, w)$ denotes the hydrodynamical pair. The momentum $P$, defined by
\begin{equation}
\label{def:P}
P(\gv) := \int_\R v w,
\end{equation}
is also conserved by the Landau-Lifshitz flow.
When $c \neq 0$, the function $\check{u}_c$ does not vanish. The hydrodynamical pair $Q_c := (v_c, w_c)$ is given by
\begin{equation}
\label{form:vc}
v_c(x) = \frac{(1 - c^2)^\frac{1}{2}}{\cosh \big( (1 - c^2)^\frac{1}{2} x \big)}, \quad {\rm and} \quad w_c(x) = \frac{c \, v_c(x)}{1 - v_c(x)^2} = \frac{c (1 - c^2)^\frac{1}{2} \cosh \big( (1 - c^2)^\frac{1}{2} x \big)}{\sinh \big( (1 - c^2)^\frac{1}{2} x \big)^2 + c^2}.
\end{equation}
The flow of \eqref{HLL} is invariant by translations and the opposite map $(v, w) \mapsto (- v, - w)$. These geometric transformations play an important role in the stability statement. We will show that the stability depends on these invariances. 

We denote   
$$Q_{c, a,s}(x) := sQ_c(x - a) := \big( s v_c(x - a),s w_c(x - a) \big),$$
for $a \in \R$ and $s\in \{\pm1\}$.
We also define
\begin{equation}
\label{def:S}
S_{\gc, \ga,\gs} := \big(V_{\gc,\ga,\gs},W_{\gc,\ga,\gs}\big) := \sum_{j = 1}^N Q_{c_j, a_j,s_j},
\end{equation}
with $N \in \N^*$, $\gc = (c_1, \ldots, c_N)$, with $c_j \neq 0$, $\ga = (a_1, \ldots, a_N) \in \R^N$ and $\gs = (s_1, \ldots, s_N) \in \{\pm1\}^N$.
In the original framework, this can be translated in the following way
$$R_{\gc,\ga,\gs} := \Big( (1 - V^2_{\gc,\ga,\gs})^{\frac{1}{2}}\cos(\Theta_{\gc,\ga,\gs}),(1 - V^2_{\gc,\ga,\gs})^{\frac{1}{2}} \sin(\Theta_{\gc,\ga,\gs}), V_{\gc,\ga,\gs}\Big),$$
where we have denoted
$$\Theta_{\gc,\ga,\gs}(x) := \int_0^x W_{\gc,\ga,\gs}(y) dy,$$
for any $x\in \R$.
In this paper, we provide the proof of the asymptotic stability around any soliton and between any two solitons of a sum of well-separated solitons with ordered speed, i.e. 
$$a_j-a_{j-1} \geq L, \quad {\rm for \ any \ } j \in \{2, \ldots , N\}, \quad {\rm where \ } L>0, \quad {\rm and } \quad c_1<\ldots<c_N.$$

Multi-solitons are orbitally stable under these invariance parameters (see \cite{DeLGr} for more details). We recall this result in the next section (see Theorem \ref{thm:mult-stab} below).

\subsection{Asymptotic stability in the original framework}
In this subsection, we provide our main result. First, we introduce a metric structure on the energy space $\boE(\R)$ in order to establish them. As it was done by de Laire and Gravejat in \cite{DeLGr}, we define the following distance
$$d_\boE (f,g):= |\check{f}(0)-\check{g}(0)| +\|f'-g'\|_{L^2(\R)} + \|f_3-g_3\|_{L^2(\R)},$$
where $f=(f_1,f_2,f_3)$ and $\check{f}=f_1+if_2$ (respectively for $g$). With this choice, $(\boE(\R),d_\boE)$ is a metric space. 
The following theorem shows the asymptotic stability around each soliton and between the solitons.
\begin{thm}
\label{thm:stabasympt_m1}
Let $\gs \in \{ \pm 1 \}^N$, $\gc^0 = (c^0_1, \ldots, c^0_N) \in (-1, 1)^N$, with $c^0_j \neq 0,$ such that 
$$c^0_1 < \ldots < c^0_N,$$ 
and $ \ga^0= (a^0_1, \ldots, a^0_N) \in \R^N.$ There exist a positive number $\beta_{\gc^0}$, depending only on $\gc^0$, and a positive number $L^0$ such that, if
$$d_\boE \big( m^0, R_{\gc^0,\ga^0,\gs} \big) \leq \beta_{\gc^0},$$
and
$$\ga^0 \in \Pos(L^0),$$
then there exist $N$ numbers $\tilde{\gc}:= \big( \tilde{c}_1,\ldots,\tilde{c}_N\big) \in (- 1, 1)^N $, with $\tilde{c}_j \neq 0$, and $2N$ functions \\$a_j \in \boC^1(\R_+, \R)$ and $\theta_j \in \boC^1(\R_+, \R)$, such that
$$a'_j(t) \to \tilde{c}_j, \quad {\rm and} \quad \theta'_j(t) \to 0,$$
as $t \to + \infty$, and for which the map
$$ m_{\theta_j}:= \Big( \cos(\theta_j)m_1-\sin(\theta_j)m_2,\sin(\theta_j)m_1+\cos(\theta_j)m_2,m_3\Big),$$
corresponding to the unique global solution $m \in \boC^0(\R, \boE(\R))$ with initial datum $m^0$, satisfies the convergences
\begin{equation}
\begin{split}
\label{conv:asymstab}
& \quad \sum_{j=1}^{N}\big[ \partial_x m_{\theta_j(t)} \big( \cdot + a_j(t), t \big) - \partial_x u_{\tilde{c}_j} \big] \rightharpoonup 0 \quad {\rm in} \ L^2(\R),\\
&\quad \sum_{j=1}^{N}\big[  m_{\theta_j(t)} \big( \cdot + a_j(t), t \big) - u_{\tilde{c}_j} \big] \to 0 \quad {\rm in} \ L^\infty_{\rm loc}(\R),\\
& \quad {\rm and } \quad \sum_{j=1}^{N}\big[ m_{3} \big( \cdot + a_j(t), t \big) - [u_{\tilde{c}_j}]_3\big] \rightharpoonup 0 \quad {\rm in} \ L^2(\R),
\end{split}
\end{equation}
as $t \to + \infty$.
In addition, for any map $b_j$ satisfying the following conditions :
\begin{equation}
\label{cond1:b_j}
\left\{
\begin{array}{r c l}
& b_1(t)< a_1(t),\\
& a_{j-1}(t) < b_j(t) < a_j(t) \quad \forall \ 2 \leq j \leq N, \\
& b_{N+1}(t) > a_N(t),
\end{array}
\right.
\end{equation}
for all $t\in \R_+$ and
\begin{equation}
\label{cond3inf:b_j}
\left\{
\begin{array}{r c l}
& \underset{t\to+\infty}{\liminf} \frac{b_j(t)}{t}>c_{j-1}^\infty ,\\ 
&\underset{t\to+\infty}{\limsup} \frac{b_j(t)}{t}< c_j^\infty,
\end{array}
\right.
\end{equation}
with
\[
\left\{
\begin{array}{r c l}
& c_0^\infty = -1,\\
& c_{N+1}^\infty = 1,
\end{array}
\right.\] 
we have
\begin{equation}
\begin{split}
\label{conv:asymstab1}
&\sum_{j=1}^{N} \partial_x m_{\theta_j(t)} \big( \cdot + b_j(t), t \big) \rightharpoonup 0 \quad {\rm in} \ L^2(\R),\\  &\sum_{j=1}^{N}\big[  m_{\theta_j(t)} \big( \cdot + b_j(t), t \big) - e_2 \big] \to 0 \quad {\rm in} \ L^\infty_{\rm loc}(\R),\\
& \quad {\rm and } \quad \sum_{j=1}^{N} m_{3} \big( \cdot + a_j(t), t \big) \rightharpoonup 0 \quad {\rm in} \ L^2(\R),
\end{split}
\end{equation}
as $t \to + \infty$, with $e_2=(0,1,0)$.
\end{thm}
The proof of this theorem is similar to the one of Theorem 1.1 in \cite{Bahri}. It relies on a modulation argument and Theorem \ref{thm:stabasympt}. The proof still applies for our case of $N$ solitons since each term of the sums in \eqref{conv:asymstab} and \eqref{conv:asymstab1} converges to zero. It remains to deal with each term separately and apply the arguments used for the case of one soliton $N$ times. In particular, \eqref{conv:asymstab} and \eqref{conv:asymstab1} are direct consequences of \eqref{conv-stab-asymp} and \eqref{conv-stab-asymp1} respectively (see Subsection 2.4 in \cite{Bahri} for more details).
\begin{remark}
The locally strong asymptotic stability result for multi-solitons, as stated by Martel, Merle and Tsai in \cite{MarMeTs1} for the KdV equation, is stronger than the two weak asymptotic stability results stated in this paper. It is still an open problem for this equation. As a matter of fact, the method used by Martel, Merle and Tsai is based on a monotonicity argument for the localized energy. This argument is not obvious in our case, since dispersion has both positive and negative speeds in contrast with the KdV case in which dispersion has only negative speeds.  
\end{remark}

\subsection{Asymptotic stability in the hydrodynamical framework}

The following theorem shows the asymptotic stability of multi-solitons in the hydrodynamical framework. We show the asymptotic stability around and between solitons.
\begin{thm}
\label{thm:stabasympt}
Let $\gc^0=(c_1^0,\ldots,c_N^0) \in (-1,1)^N$, with $c_j^0 \neq 0$ for all $j=1,\ldots,N$, such that 
there exist $L_0,\alpha_0>0$ with the following properties. Given any $(v_0, w_0) \in X(\R)$, there exist  $L>L_0$ and $\alpha < \alpha_0$ such that if $(v_0, w_0) \in \boV(\alpha,L)$, then there exist $\ga := \big( a_1,\ldots,a_N\big) \in \boC^1(\R_+, \R^N)$, $\gc := \big( c_1,\ldots,c_N\big) \in \boC^1(\R_+, (-1,1)\setminus \{0\} ^N)$ and non zero different speeds $\gc^{+\infty}= \big( c_1^{+\infty},\ldots,c_N^{+\infty} \big)$ such that the unique global solution $(v, w) \in \boC^0(\R, \boN\boV(\R))$ to \eqref{HLL} with initial datum $(v_0, w_0)$ satisfies, for all $j \in \{1,\ldots,N\}$,
\begin{equation}
\label{conv-stab-asymp}
\varepsilon(t,.+a_j(t)):=(v,w)(t,x+a_j(t))-\sum_{k=1}^N Q_{c_k(t)}(x+a_j(t)-a_k(t)) \rightharpoonup 0 \quad {\rm in} \ X(\R),
\end{equation}
as well as
\begin{equation}
\label{conv-stab-asymp1}
\varepsilon(t,.+b_j(t)):=(v,w)(t,x+b_j(t))-\sum_{k=1}^N Q_{c_k(t)}(x+b_j(t)-a_k(t)) \rightharpoonup 0 \quad {\rm in} \ X(\R),
\end{equation}
for any $\gb := \big( b_1,\ldots,b_{N+1}\big) \in \boC^1(\R_+, \R^{N+1})$ with $b_j$ satisfying \eqref{cond1:b_j} and 
\begin{equation}
\label{cond3:b_j}
c_{j-1}^\infty < \lim_{t\to +\infty} b'_j(t) < c_j^\infty.
\end{equation}
Moreover, we have
\begin{equation}
\label{conv-parameters}
c_j(t)\to c_j^{+\infty},\quad a'_j(t) \to c_j^{+\infty},
\end{equation}
as $t \to + \infty$.
\end{thm}

In fact, all the solitons in \eqref{conv-stab-asymp} with speed $c_k$ for $k \neq j$ are weakly convergent to $0$ in $X(\R)$ as $t \to + \infty ,$ due to \eqref{conv-parameters}, so that \eqref{conv-stab-asymp} truly provides the asymptotic stability of the soliton with speed $c_j$. For \eqref{conv-stab-asymp1}, all the solitons are weakly convergent to $0$ in $X(\R)$ as $t \to + \infty ,$ so that \eqref{conv-stab-asymp1} provides the asymptotic stability of the zero solution between the solitons.
\begin{remark}
$(i)$ For \eqref{conv-stab-asymp1}, we begin by proving the convergence for $\gb := \big( b_1,\ldots,b_{N+1}\big) \in \boC^1(\R_+, \R^{N+1})$ with $b_j$ satisfying \eqref{cond1:b_j} and \eqref{cond3:b_j}. Then, we show that it remains also true for any $b_j$ verifying \eqref{cond3inf:b_j} in order to deduce \eqref{conv:asymstab1} (see the end of Subsection 4.1 for the proof).

$(ii)$The case when $c_j^0 \neq 0$ is excluded from the statement. In fact, we cannot use the hydrodynamical formulation in that case because the solitons can vanish. In addition, the Liouville type theorem cannot be applied as well as the orbital stability theorem. To our knowledge, this is still an open problem.
\end{remark}
The proof relies on the strategy developed by Martel, Merle and Tsai in \cite{MarMeTs1}.

\subsection{Plan of the paper}
In the second section, we recall the orbital stability result for the multi-solitons, stated by de Laire and Gravejat in \cite{DeLGr}, which is an important tool to prove our results.

In the third section, we prove the asymptotic stability around solitons. More precisely, we show that any solution close to the sum of $N$ solitons is weakly convergent to a soliton in the translating neighbourhood of each soliton. We state that all other solitons stay far in the way that in this region the problem reduces to the asymptotic stability for a single soliton. This is the reason why we can use the Liouville type theorem proved in \cite{Bahri}.

In the last section, we change the translation parameter to show that any solution, corresponding to an initial datum close to the sum of $N$ solitons, converges weakly to zero when it is moving in the core of the region separating two solitons. For this, we establish a Liouville type theorem, which affirms that small solutions which are smooth and exponentially localized are zero solutions. As a consequence, \eqref{conv-stab-asymp1} claims that there is no interaction between well separated solitons with ordered speed.

\section{Orbital stability in the hydrodynamical framework}
In this section, we first recall the orbital stability result proved by de Laire and Gravejat in \cite{DeLGr}. In order to quantify it precisely, we set 
$$\boN\boV(\R) := \Big\{ \gv = (v, w) \in H^1(\R) \times L^2(\R), \ {\rm s.t.} \ \max_\R |v| < 1 \Big\}.$$
In the sequel we consider this space as a metric space equiped with the metric structure provided by the norm
$$\| \gv \|_{H^1 \times L^2} := \Big( \| v \|_{H^1}^2 + \| w \|_{L^2}^2 \Big)^\frac{1}{2}.$$
\begin{thm}\cite{DeLGr}
\label{thm:mult-stab}
Let $\gs^* \in \{ \pm 1 \}^N$ and $\gc^* = (c_1^*, \ldots, c_N^*) \in (-1, 1)^N$, with $c_j^* \neq 0$, such that
\begin{equation}
\label{eq:order-speed}
c_1^* < \ldots < c_N^*.
\end{equation}
There exist positive numbers $\alpha^*$, $L^*$ and $A^*$, depending only on $\gc^*$ such that, if $\gv^0 \in \boN\boV(\R)$ satisfies the condition
\begin{equation}
\label{def:alpha0}
\alpha^0 := \big\| \gv^0 - S_{\gc^*, \ga^0, \gs^*} \big\|_{H^1 \times L^2} \leq \alpha^*,
\end{equation}
for points $\ga^0 = (a_1^0, \ldots, a_N^0) \in \R^N$ such that
$$L^0 := \min \big\{ a_{j + 1}^0 - a_j^0, 1 \leq j \leq N - 1 \big\} \geq L^*,$$
then the solution $\gv$ to \eqref{HLL} with initial datum $\gv^0$ is globally well-defined on $\R _+$, and there exists a function $\ga = (a_1, \ldots, a_N) \in \boC^1(\R_+, \R^N)$ such that
\begin{equation}
\label{est:a'}
\sum_{j = 1}^N \big| a_j'(t) - c_j^* \big| \leq A^* \Big( \alpha^0 + \exp \Big( - \frac{\nu_{\gc^*} L^0}{65} \Big) \Big),
\end{equation}
and
\begin{equation}
\label{Goal}
\big\| \gv(\cdot, t) - S_{\gc^*, \ga(t), \gs^*} \big\|_{H^1 \times L^2} \leq A^* \Big( \alpha^0 + \exp \Big( - \frac{\nu_{\gc^*} L^0}{65} \Big) \Big),
\end{equation}
for any $t \in \R_+$.
\end{thm}
Given a positive number $L > 0$, we introduce the set of well-separated and ordered positions
$$\Pos(L) := \big\{ \ga = (a_1, \ldots, a_N) \in \R^N, \ {\rm s.t.} \ a_{j + 1} > a_j + L \ {\rm for} \ 1 \leq j \leq N - 1 \big\},$$
and we set
$$\boV(\alpha, L) := \Big\{ \gv = (v, w) \in H^1(\R) \times L^2(\R), \ {\rm s.t.} \ \inf_{\ga \in \Pos(L)} \big\| \gv - S_{\gc^*, \ga,\gs^*} \big\|_{H^1 \times L^2} < \alpha \Big\},$$
for $\alpha > 0$. We also define
$$\mu_\gc := \min_{1 \leq j \leq N} |c_j|, \quad {\rm and} \quad \nu_\gc := \min_{1 \leq j \leq N} \big( 1 - c_j^2 \big)^\frac{1}{2},$$
for any $\gc \in (- 1, 1)^N$. The following proposition provides some details contained in the proof of Theorem \ref{thm:mult-stab}. In particular, it shows the existence of the speed and the translation parameters for each soliton (see \cite{DeLGr} for the proof). It is an important tool for the proof of the asymptotic stability result.
\begin{prop}\cite{DeLGr}
\label{prop:modul}
There exist positive numbers $\alpha_1^*$ and $L_1^*$, depending only on $\gc^*$ and $\gs^*$, such that we have the following properties.\\
$(i)$ Any pair $\gv = (v , w) \in \boV(\alpha_1^*, L_1^*)$ belongs to $\boN\boV(\R)$, with
\begin{equation}
\label{eq:min-v}
1 - v^2 \geq \frac{1}{8} \mu_{\gc^*}^2.
\end{equation}
$(ii)$ There exist two maps $\gc \in \boC^1(\boV(\alpha_1^*, L_1^*), (- 1, 1)^N)$ and $\ga \in \boC^1(\boV(\alpha_1^*, L_1^*), \R^N)$, and a positive number $A^*$, depending only on $\gc^*$ and $\gs^*$, such that, if
 $$\big\| \gv - S_{\gc^*, \ga^*, \gs^*} \big\|_{H^1 \times L^2} < \alpha,$$
for $\ga^* \in \Pos(L)$, with $L > L_1^*$ and $\alpha < \alpha_1^*$, then we have
\begin{equation}
\label{eq:est-beps}
\| \varepsilon \|_{H^1 \times L^2} + \sum_{j = 1}^N \big| c_j(\gv) - c_j^* \big| + \sum_{j = 1}^N \big| a_j(\gv) - a_j^* \big| \leq A^* \Big( \alpha + \exp \Big( - \frac{\nu_{\gc^*} L}{32} \Big) \Big),
\end{equation}
as well as
\begin{equation}
\label{eq:est-coef}
\ga(\gv) \in \Pos(L - 1), \quad \mu_{\gc(\gv)} \geq \frac{1}{2} \mu_{\gc^*} \quad {\rm and} \quad \nu_{\bc(\gv)} \geq \frac{1}{2} \nu_{\gc^*},
\end{equation}
where
$$\varepsilon = \gv - S_{\gc(\gv), \ga(\gv), \gs^*},$$
satisfies the orthogonality conditions
\begin{equation}
\label{eq:ortho}
\langle \eps, \partial_x Q_{c_k(\gv)} \rangle_{L^2(\R)^2} = \langle \eps, \chi_{c_k(\gv)} \rangle_{L^2(\R)^2} = 0,
\end{equation}
for any $k\in \{1,\ldots ,N\}.$ The function $\chi_{c_k(\gv)}$ stands here for an eigenvector of the quadratic form $\boH_{c_k(\gv)}:= E''(Q_{c_k(\gv)})-c_k(\gv)P''(Q_{c_k(\gv)})$ associated to its unique negative eigenvalue. 
\end{prop}
\begin{remark}
The second orthogonality condition in \eqref{eq:ortho} is not the same as the one used by de Laire and Gravejat in \cite{DeLGr}. However, the result remains true by the same argument used in \cite{Bahri} (see Section 3 in \cite{Bahri} for more details). Moreover, we need this orthogonality condition in order to apply the Liouville type theorem (Theorem \ref{th:liouville} below) (see Subsection 2.3.3 in \cite{Bahri} for more details).  
\end{remark}
Next, we recall the result for only one soliton which is a direct consequence of Theorem \ref{thm:mult-stab}. It is an important tool for the proof of \eqref{conv:asymstab} since we analyse the soliton around each soliton.
\begin{thm}\cite{DeLGr}
\label{thm:orbistab}
Let $c \in (- 1, 1) \setminus \{ 0 \}$. There exists a positive number $\alpha_c$, depending only on $c$, with the following properties. Given any $(v_0,w_0) \in \boN\boV(\R)$ such that
\begin{equation}
\label{cond:alpha}
\alpha_0 := \big\| (v_0,w_0) - Q_{c, a} \big\|_{X(\R)} \leq \alpha_c,
\end{equation}
for some $a \in \R$, there exist a unique global solution $(v,w) \in \boC^0(\R, \boN\boV(\R))$ to \eqref{HLL} with initial datum $(v_0,w_0)$, two maps $c \in \boC^1(\R, (- 1, 1) \setminus \{ 0 \})$ and $a \in \boC^1(\R, \R)$, and two positive numbers $\sigma_c$ and $A_c$, depending only and continuously on $c$, such that 
\begin{equation}
\label{eq:max-v}
\max_{x \in \R} v(x, t) \leq 1 - \sigma_c,
\end{equation}
\begin{equation}
\label{eq:modul0}
\big\| \eps(\cdot, t) \big\|_{X(\R)} + \big| c(t) - c \big| \leq A_c \alpha^0,\\
\end{equation}
and
\begin{equation}
\label{eq:modul1}
\big| c'(t) \big| + \big| a'(t) - c(t) \big| \leq A_\gc \big\| \eps(\cdot, t) \big\|_{X(\R)},
\end{equation}
for any $t\in \R$, where the function $\eps$ is defined by
\begin{equation}
\label{def:eps}
\eps(\cdot, t) := \big( v(\cdot + a(t), t), w(\cdot + a(t), t) \big) - Q_{c(t)},
\end{equation}
and satisfies the orthogonality conditions
\begin{equation}
\label{eq:ortho2}
\langle \eps(\cdot,t), \partial_x Q_{c(t)} \rangle_{L^2(\R)^2} = \langle \eps(\cdot,t), \chi_{c(t)} \rangle_{L^2(\R)^2} = 0,
\end{equation}
for any $t \in \R$.
\end{thm}

Set
$$\gc(t) := \gc(\gv(\cdot, t)) := \big( c_1(t), \ldots, c_N(t) \big) \quad {\rm and} \quad \ga(t) := \ga(\gv(\cdot, t)) := \big( a_1(t), \ldots, a_N(t) \big),$$
as well as
\begin{equation}
\label{def:beps-t}
\varepsilon(\cdot, t) := \big( \varepsilon_1(\cdot, t), \varepsilon_2(\cdot, t) \big) = \gv(\cdot, t) - S_{\gc(t), \ga(t), \gs^*}.
\end{equation}
The pair $\varepsilon$ is well defined and satisfies the orthogonality conditions
\begin{equation}
\label{eq:ortho1}
\langle \eps(\cdot,t), \partial_x Q_{c_k(t)} \rangle_{L^2(\R)^2} = \langle \eps(\cdot,t), \chi_{c_k(t)} \rangle_{L^2(\R)^2} = 0,
\end{equation}
for any $t \in \R_+$ and for any $k\in \{1,\ldots ,N\}$ (see \cite{DeLGr} for more details). For $\alpha$ and $L$ given by Proposition~\ref{prop:modul}, we also infer from the results in \cite{DeLGr} that 
\begin{equation}
\label{eq:est-beps-t}
\| \varepsilon(\cdot, t) \|_{H^1 \times L^2} + \sum_{j = 1}^N \big| c_j(t) - c_j^* \big| \leq A^* \Big( \alpha + \exp \Big( - \frac{\nu_{\gc^*} L}{65} \Big) \Big),
\end{equation}
and
\begin{equation}
\label{eq:est-coef-t}
\ga(t) \in \Pos(L - 1), \quad \mu_{\gc(t)} \geq \frac{1}{2} \mu_{\gc^*} \quad {\rm and} \quad \nu_{\gc(t)} \geq \frac{1}{2} \nu_{\gc^*}.
\end{equation}
\section{Asymptotic stability around the solitons in the hydrodynamical variables}
\subsection{Proofs of \eqref{conv-stab-asymp} and \eqref{conv-parameters} }

Let $\gc^0$ be as in Theorem \ref{thm:stabasympt} and $\gv_0$ be any pair which belongs to the set $\boV(\alpha,L)$ with $\alpha$ and $L$ as in the hypothesis of Theorem \ref{thm:stabasympt}. 

Let $j \in \{1,\ldots,N\}$. 
By \eqref{eq:est-beps-t}, the functions $\varepsilon$ and $c_j$ are uniformly bounded in $X(\R)$, respectively in $\R$. Then, there exist $\tilde{\varepsilon}_{j,0} \in X(\R)$\footnote{In view of \eqref{eq:est-beps-t}, the norm of $\tilde{\varepsilon}_{j,0}$ in $X(\R)$ is small.} and $\tilde{c}_{j,0} \in (-1,1)\setminus \{0\}$ such that, up to a subsequence,
\begin{equation}
\label{conv:eps}
\varepsilon(t_n,.+a_j(t_n)) \rightharpoonup \tilde{\varepsilon}_{j,0} \quad \hbox{in} \ \ X(\R)  \quad \hbox{and} \quad c_j(t_n)\to \tilde{c}_{j,0} \quad \hbox{as} \quad n\to +\infty.
\end{equation}
Indeed, the bounds in \eqref{eq:est-beps-t} and the possibility to choose $\alpha$ small enough guarantee that $\tilde{c}_{j,0}$ stays always close to $c_j^0$ which prevents $\tilde{c}_{j,0}$ to be in $\{-1,0,1\}$ for any $j\in\{1,\ldots,N\}$.

We set $\tilde{\gv}_{j,0}=(\tilde{v}_{j,0},\tilde{w}_{j,0}):= Q_{\tilde{c}_{j,0}}+\tilde{\varepsilon}_{j,0}$ and denote by $\tilde{\gv}_j=(\tilde{v}_j,\tilde{w}_j)$ the unique global solution to \eqref{HLL} corresponding to this initial datum $\tilde{\gv}_{j,0}$. We claim that this solution exponentially decays with respect to the space variable for any time, as well as all its space derivatives. More precisely, we have
\begin{prop}
\label{prop:smooth}
The pair $(\tilde{v}_j,\tilde{w}_j)$ is indefinitely smooth and exponentially decaying on $\R \times\R$. Moreover, given any $k \in \N$, there exist a positive constant $A_{k, \gc}$, depending only on $k$ and $\gc$, and a function $\tilde{a}_j \in \boC^1(\R, \R)$ such that
\begin{equation}
\label{eq:smooth}
\int_\R \big[ (\partial_x^{k + 1} \tilde{v}_j)^2 + (\partial^k_x \tilde{v}_j)^2 + (\partial_x^k \tilde{w}_{j})^2 \big](x + \tilde{a}_j(t), t) \exp\big( \frac{\nu_\gc}{16} |x|\big) \, dx \leq A_{k, \gc},
\end{equation}
for any $t\in \R$.
\end{prop}
With this proposition at hand, we can finish the proof of \eqref{conv-stab-asymp}. 
We recall the Liouville type theorem stated in \cite{Bahri}.
\begin{thm}\cite{Bahri}
\label{th:liouville}
Let $j \in \{1,\ldots,N\}$, $c_j \in (-1,1)\setminus \{0\}$ and $(\tilde{v}_j,\tilde{w}_{j})$ a solution of \eqref{HLL} satisfying \eqref{eq:smooth} and 
\begin{equation}
\label{cond:alpha1}
\|(\tilde{v}_{j,0},\tilde{w}_{j,0})-Q_{c_j}\|_{X(\R)} \leq \alpha . 
\end{equation}

Then, there exist two numbers $x^*\in \R$ and $c^* \in (-1,1)\setminus \{0\}$ such that
$$(\tilde{v}_j,\tilde{w}_{j})(t,x)=Q_{c^*}(x-x^*-c^* t) \quad \forall (t,x) \in \R\times \R .$$
\end{thm}
Due to the orbital stability of $Q_{\tilde{c}_{j,0}},$ condition \eqref{cond:alpha1} is satisfied when $\alpha_0$ is small enough. Applying Theorem \ref{th:liouville}, we get $x^*\in \R$ and $c^* \in (-1,1)\setminus \{0\}$ such that we have 
$$\tilde{\gv}_j(t,x)=Q_{c^*}(x-x^*-c^* t), \quad \forall (t,x) \in \R\times \R .$$
In particular, we have $ Q_{\tilde{c}_{j,0}}(x)+\tilde{\varepsilon}_{j,0}(x)= Q_{c^*}(x-x^*).$ We claim that $x^*=0$. Indeed, we use the fact that $\|\tilde{\varepsilon}_{j,0}\|_{X(\R)} \leq \alpha$ and a modulation argument to obtain $|c^*-\tilde{c}_{j,0}| \leq A_\gc \alpha$ and $|x^*| \leqslant A_\gc \alpha .$ We define 
$$ h(c^*,x^*) = \int_{\R} \big< Q_{c^*}(x-x^*), Q'_{\tilde{c}_{j,0}} \big> .$$
We have $$ \partial_{x^*} h(\tilde{c}_{j,0},0) = - \int_{\R} |Q'_{\tilde{c}_{j,0}}|^2 \neq 0.$$
From the implicit function theorem, there exist a neighbourhood $V$ of $(\tilde{c}_{j,0},0)$ and a function $\phi$ such that $(c^*,x^*) \in V$ and $ h(c^*,x^*)=0$ if and only if $x^* = \phi(c^*)$. Since, by parity, $h(c^*,0)=0$, we infer that $x^*=0$.

Next, we set $g(c^*)=\int_{\R} \big< Q_{c^*}-Q_{\tilde{c}_{j,0}},Q_{\tilde{c}_{j,0}} \big> $. Since $g'(\tilde{c}_{j,0}) \neq 0,$ we can prove that $c^*=\tilde{c}_{j,0}$, which leads to the fact that $\tilde{\varepsilon}_0 \equiv 0$. This allows us to deduce the convergence \eqref{conv-stab-asymp} for a subsequence of $(t_n)_{n \in \N}$. 

Finally, we prove \eqref{conv-stab-asymp} and \eqref{conv-parameters} for $t \to + \infty$. Since $a_l(t_{n_k})-a_j(t_{n_k}) \to \infty$ for all $l\neq j$ , the solution converges to only one soliton because the other solitons converges to zero. This means that we have 
$$\big( v(\cdot + a_j(t_{n_k}), t_{n_k}), w(\cdot + a_j(t_{n_k}), t_{n_k}) \big) - Q_{c_j(t_{n_k})} \rightharpoonup 0 \quad {\rm in} \ X(\R),$$
as $k \to + \infty$. This restricts the problem to the case of only one soliton. The proof is then similar to the one stated by Béthuel, Gravejat and Smets in \cite{BetGrSm2}. It relies on the monotonicity formula for the quantities $\boI_{j,y_0}$ in Proposition \ref{prop:evol-momentum}.

The main idea is to show that $\tilde{c}_{j,0}$ is independent of the sequence $(t_n)_{n \in \N}$. Assume by contradiction that for two different sequences $(t_n)_{n \in \N}$ and $(s_n)_{n \in \N}$, both tending to $+ \infty$, we have
$$c_j(t_n) \to c_{j,1} \quad {\rm and} \quad c_j(s_n) \to c_{j,2},$$
as $n \to + \infty$, with $c_{j,1}\neq c_{j,2}$ satisfying \eqref{eq:est-beps-t}. In addition, we suppose that we have
\begin{equation}
\label{eq:cona}
\big( v(\cdot + a(t_n), t_n), w(\cdot + a(t_n), t_n) \big)-Q_{c_j(t_n)} \rightharpoonup 0 \quad {\rm in} \ X(\R),
\end{equation}
and 
\begin{equation}
\label{eq:conb}
\big( v(\cdot + a(s_n), s_n), w(\cdot + a(s_n), s_n) \big)-Q_{c_j(s_n)} \rightharpoonup 0 \quad {\rm in} \ X(\R).
\end{equation}
Note that these two convergences are different since $Q_{c_j(t_n)}\to Q_{c_{j,1}}$ and $Q_{c_j(s_n)}\to Q_{c_{j,2}}$ as $n\to \infty$. We may assume, without loss of generality, that $c_{j,1} < c_{j,2}$ and that the sequences $(t_n)_{n\in \N}$ and $(s_n)_{n\in \N}$ are strictly increasing and are taken such that
\begin{equation}
\label{eq:imbriques}
t_n + 1 \leq s_n \leq t_{n + 1} - 1,
\end{equation}
for any $n \in \N$. Let $\delta > 0$. For $y_0$ sufficiently large, we can define the quantities $\boI_{j,y_0}$ as in \eqref{def:boIj}, and deduce from \eqref{eq:imbriques} and \eqref{eq:monobisI} that
\begin{equation}
\label{eq:monoencore}
\boI_{j,\pm y_0}(s_n) \geq \boI_{j,\pm y_0}(t_n) - \frac{\delta}{10} \quad {\rm and} \quad \boI_{j,\pm y_0}(t_{n + 1}) \geq \boI_{j,\pm y_0}(s_n) - \frac{\delta}{10},
\end{equation}
for any $n \in \N$. On the other hand, by \eqref{eq:cona} and \eqref{eq:conb}, there exists an integer $n_0$ such that
\begin{equation}
\label{eq:petitpetit}
\big| \boI_{j,- y_0}(t_n) - \boI_{j, y_0}(t_n) - P(Q_{c_j(t_n)}) \big| \leq \frac{\delta}{5},
\end{equation}
and
\begin{equation}
\label{eq:petitpetitpetit}
\big| \boI_{j,- y_0}(s_n) - \boI_{j, y_0}(s_n) - P(Q_{c_j(s_n)}) \big| \leq \frac{\delta}{5},
\end{equation}
for any $n \geq n_0$ and for $y_0$ large enough. From \eqref{eq:monoencore}, \eqref{eq:petitpetit} and \eqref{eq:petitpetitpetit}, we have
$$\boI_{j, y_0}(s_n) \geq \boI_{j, y_0}(t_n) + \frac{\delta}{2},$$
for any $n \geq n_0$, this yields, using  \eqref{eq:monoencore} again, that
$$\boI_{j, y_0}(t_{n + 1}) \geq \boI_{j, y_0}(t_n) + \frac{2 \delta}{5},$$
for any $n \geq n_0$. Therefore, the sequence $(\boI_{j, y_0}(t_n))_{n \in \N}$ is unbounded, which leads to a contradiction with the fact that the pair $(v,w)$ has a bounded energy.

The second convergence in \eqref{conv-parameters} follows from the fact that
$$a_j(t_n + t) - a_j(t_n) \to c_j^{+\infty} t,$$
for any fixed $t \in \R$ and any sequence $(t_n)_{n \in \N}$ tending to $+ \infty$ (due to \eqref{sologne2}), and Lemma 2 in \cite{BetGrSm2} (see \cite{BetGrSm2} for more details). \qed

\subsection{Localization and smoothness of the limit profile}
In this section, we prove Proposition \ref{prop:smooth}. First, we use \eqref{est:a'} and \eqref{eq:est-beps-t} to claim that
\begin{equation}
\label{eq:lvmh1}
\min_{j=1,\ldots,N} \big\{ c_j(t)^2, a_j'(t)^2 \big\} \geq \frac{\mu_\gc^2}{2}, \qquad \max_{j=1,\ldots,N} \big\{c_j(t)^2, a_j'(t)^2 \big\} \leq 1 + \frac{\mu_\gc^2}{2}, 
\end{equation}
and
\begin{equation}
\label{eq:lvmh2}
\big\| V_{\gc,\ga(t),\gs} - v( t) \big\|_{L^\infty(\R)} \leq \min \Big\{ \frac{\mu_\gc^2}{4}, \frac{\nu_\gc^2}{16} \Big\},
\end{equation}
for any $t\in \R$. In particular, we conclude that $\tilde{c}_{j,0} \in (- 1, 1) \setminus \{ 0 \}$, so that $Q_{\tilde{c}_{j,0}}$ is a dark soliton.

In addition, for $j\in \{1,\ldots,N\}$, we have 
\begin{equation}
\label{eq:encorezero}
\big| \tilde{c}_{j,0} - c_j \big| \leq A_{\mu_\gc} \alpha.
\end{equation}
On the other hand, by the weak lower semi-continuity of the norm, \eqref{eq:est-beps-t} and \eqref{conv:eps}, we infer that
\begin{equation}
\label{eq:encoreune}
\big\| (\tilde{v}_{j,0}, \tilde{w}_{j,0}) - Q_{c_j} \big\|_{X(\R)} \leq A_{\mu_\gc} \alpha + \big\| Q_{c_j} - Q_{\tilde{c}_{j,0}} \big\|_{X(\R)} \leq A_{\mu_\gc} \alpha .
\end{equation}

Now, we suppose that $\alpha$ is sufficiently small so that, by \eqref{eq:encoreune}, 
\begin{equation}
\label{cond:alpha*}
\big\| (\tilde{v}_{j,0}, \tilde{w}_{j,0}) - Q_{c_j} \big\|_{X(\R)} \leq \alpha_\gc.
\end{equation}
By Theorem \ref{thm:orbistab}, there exist two maps $\tilde{c}_j \in \boC^1(\R, (- 1, 1) \setminus \{ 0 \})$ and $\tilde{a}_j \in \boC^1(\R, \R)$ such that the function $\tilde{\eps}_j$ defined by
\begin{equation}
\label{def:eps*}
\tilde{\eps}_j(\cdot, t) := \big( \tilde{v}_j(\cdot + \tilde{a}_j(t), t), \tilde{w}_{j}(\cdot + \tilde{a}_j(t), t) \big) - Q_{\tilde{c}_j(t)},
\end{equation}
satisfies the estimates
\begin{equation}
\label{eq:modul0bis}
\big\| \tilde{\eps}_j(\cdot, t) \big\|_{X(\R)} + \big| \tilde{c}_j(t) - c_j \big|+ \big| \tilde{a}'_j(t) - \tilde{c}_j(t) \big| \leq A_\gc \big\| (\tilde{v}_{j,0}, \tilde{w}_{j,0}) -Q_{c_j} \big\|_{X(\R)},
\end{equation}
and the orthogonality conditions
\begin{equation}
\label{eq:ortho3}
\langle \tilde{\eps}_j(\cdot,t), \partial_x Q_{\tilde{c}_j(t)} \rangle_{L^2(\R)^2} = \langle \tilde{\eps}_j(\cdot,t), \chi_{\tilde{c}_j(t)} \rangle_{L^2(\R)^2} = 0,
\end{equation}
for any $t\in \R$.

Using \eqref{eq:encoreune} and \eqref{eq:modul0bis}, and choosing $\alpha$ small enough we claim that
\begin{equation}
\label{eq:lvmh1bis}
\min \big\{ \tilde{c}_j(t)^2, \tilde{a}_j'(t)^2 \big\} \geq \frac{\mu_\gc^2}{4}, \qquad \max \big\{\tilde{c}_j(t)^2, \tilde{a}_j'(t)^2  \big\} \leq 1 + \mu_\gc^2, 
\end{equation}
and 
\begin{equation}
\label{eq:lvmh2bis}
\big\| v_{c_j}(\cdot) - \tilde{v}_j(\cdot + \tilde{a}_j(t), t) \big\|_{L^\infty(\R)} \leq \min \Big\{ \frac{\mu_\gc^2}{4}, \frac{1 - \mu_\gc^2}{16} \Big\},
\end{equation}
for any $t\in \R$.
We then prove the following weak continuity property in the hydrodynamical framework.
\begin{prop}
\label{prop:reprod}
Let $j \in \{ 1, \ldots ,N \}$ and $t \in \R$ be fixed. Then,
\begin{equation}
\label{sologne1}
( v,w)(\cdot + a_j(t_n), t_n + t) \rightharpoonup (\tilde{v}_j, \tilde{w}_{j})(\cdot, t) \big) \quad {\rm in} \ X(\R),
\end{equation}
while
\begin{equation}
\label{sologne2}
a_j(t_n + t) - a_j(t_n) \to \tilde{a}_j(t), \quad {\rm and} \quad c_j(t_n + t) \to \tilde{c}_j(t),
\end{equation}
as $n \to + \infty$. In particular, we have
\begin{equation}
\label{sologne3}
(v,w)(\cdot + a_j(t_n+t) , t_n + t) \rightharpoonup (\tilde{v}_j,\tilde{w}_j)(\cdot+\tilde{a}_j(t), t) \quad {\rm in} \ X(\R),
\end{equation}
as $n \to + \infty$.
\end{prop}
The weak continuity of the flow and of the modulation parameters were proved in \cite{Bahri} in the case of a simple soliton. The proof of Proposition \ref{prop:reprod} is similar. 
\begin{proof}
Let $j \in \{1,\ldots,N\}$ be a fixed integer. 
First, we prove \eqref{sologne1}. By the second convergence in \eqref{conv:eps} and the explicit formula of $Q_{c_j(t_n)}$ in \eqref{form:vc}, we can infer that
$$Q_{c_j(t_n)} \to Q_{\tilde{c}_{j,0}} \quad {\rm in} \ X(\R),$$
as $n \to + \infty$. This leads, using the first convergence in \eqref{conv:eps}, to
$$\big( v(\cdot + a_j(t_n), t_n), w(\cdot + a_j(t_n), t_n) \big) \rightharpoonup \tilde{\eps}_{j,0} + Q_{\tilde{c}_{j,0}} \quad {\rm in} \ X(\R),$$
as $n \to + \infty$. In view of the fact that $t \mapsto (v(\cdot + a_j(t_n), t_n + t), w(\cdot + a_j(t_n), t_n + t))$ and $(\tilde{v}_j, \tilde{w}_j)$ are the solutions to \eqref{HLL} with initial data $(v(\cdot + a_j(t_n), t_n), w(\cdot + a_j(t_n), t_n))$, respectively $\eps_0^* + Q_{c_0^*}$, we deduce \eqref{sologne1} from the weak continuity of the flow (see Proposition A.1 in \cite{Bahri} for more details.)

Next, let us prove \eqref{sologne2}. By \eqref{eq:modul0} and \eqref{eq:modul1} the maps $a'_j$ and $ c_j$ are bounded on $\R$, so that the sequences $(a_j(t_n + t) - a_j(t_n))_{n \in \N}$ and $(c_j(t_n + t))_{n \in \N}$ are bounded. Hence it is sufficient to prove that the unique possible accumulation points for these sequences are $\tilde{a}_j(t)$, respectively $\tilde{c}_j(t)$.

We suppose now that, up to a possible subsequence, we have
\begin{equation}
\label{beuvron}
a_j(t_n + t) - a_j(t_n) \to \alpha_j, \quad {\rm and} \quad c_j(t_n + t) \to \sigma_j,
\end{equation}
as $n \to + \infty$. Given a function $\phi \in H^1(\R)$, we write
\begin{align*}
&\big\langle v(\cdot + a_j(t_n + t), t_n + t), \phi \big\rangle_{H^1(\R)} \\ &=\big\langle v(\cdot + a_j(t_n), t_n + t), \phi(\cdot - a_j(t_n + t) + a_j(t_n)) - \phi(\cdot - \alpha_j) \big\rangle_{H^1(\R)}\\
& + \big\langle v(\cdot + a_j(t_n), t_n + t), \phi(\cdot - \alpha_j) \big\rangle_{H^1(\R)}.
\end{align*}
Since we know that
$$\phi(\cdot + h) \to \phi \quad {\rm in} \ H^1(\R),$$
when $h \to 0$, we can use \eqref{sologne1} and \eqref{beuvron} to infer that
$$v(\cdot + a_j(t_n + t), t_n + t) \rightharpoonup \tilde{v}_j(\cdot + \alpha_j, t) \quad {\rm in} \ H^1(\R),$$
as $n \to + \infty$. Similarly, we obtain
$$w(\cdot + a_j(t_n + t), t_n + t) \rightharpoonup \tilde{w}_j(\cdot + \alpha_j, t) \quad {\rm in} \ L^2(\R).$$
By \eqref{beuvron} we also have
$$Q_{c_j(t_n + t)} \to Q_{\sigma_j} \quad {\rm in} \ X(\R),$$
as $n \to + \infty$. This leads to
\begin{equation}
\label{neung}
\eps(\cdot, t_n + t) \rightharpoonup \big( \tilde{v}_j(\cdot + \alpha_j, t), \tilde{w}_j(\cdot + \alpha_j, t) \big) - Q_{\sigma_j} \quad {\rm in} \ X(\R),
\end{equation}
as $n \to + \infty$.

Now, we use the fact that the function $\chi_c$ is continuous with respect to the parameter $c$, \eqref{form:vc} and the second convergence in \eqref{beuvron} to prove that
$$\partial_x Q_{c_j(t_n + t)} \to \partial_x Q_{\sigma_j} \quad \hbox{and} \quad \chi_{c_j(t_n + t)} \to \chi_{\sigma_j} \quad {\rm in} \ L^2(\R)^2,$$
as $n \to + \infty$. Combining this with \eqref{neung}, we can take the limit $n \to + \infty$ in the two orthogonality conditions in \eqref{eq:ortho3} to obtain
$$\big\langle (\tilde{v}_j(\cdot + \alpha_j, t), \tilde{w}_j(\cdot + \alpha_j, t)) - Q_{\sigma_j}, \partial_x Q_{\sigma_j} \big\rangle_{L^2(\R)^2} = \big\langle (\tilde{v}_j(\cdot + \alpha_j, t), \tilde{w}_j(\cdot + \alpha_j, t)) - Q_{\sigma_j},  \chi_{\sigma_j} \big\rangle_{L^2(\R)^2} = 0.$$
Since the parameters $\tilde{a}_j(t)$ and $\tilde{c}_j(t)$ are uniquely defined in \eqref{def:eps*}, we infer that
\begin{equation}
\label{romorantin}
\alpha_j = \tilde{a}_j(t), \quad {\rm and} \quad \sigma_j = \tilde{c}_j(t),
\end{equation}
which is enough to complete the proof of \eqref{sologne2}. Convergence \eqref{sologne3} follows combining \eqref{def:eps*} with \eqref{neung} and \eqref{romorantin}.
\end{proof}
 
Now, we consider the function $\Phi$, which is defined on $\R$ by 
\begin{equation}
\label{eq:defiphi}
\Phi(x) := \frac{1}{2} \Big( 1 + \tanh \big( \frac{\nu_{\gc}}{16} x \big) \Big).
\end{equation}
Recall that $\Phi'$ verifies the following property
\begin{equation}
\label{rhea2}
\big| \Phi'''(x) \big| \leq \frac{\nu_{\gc}^2}{64} \Phi'(x) \leq \frac{\nu_{\gc}^3}{512} \, \exp \Big( - \frac{\nu_{\gc} }{16} | x | \Big).
\end{equation}
We set
$$\delta_\gc := \frac{1}{2} \min \{ 1+c_1, c_2-c_1, c_3-c_2, \ldots , c_N-c_{N-1},1-c_N\}$$
for any $\gc \in (- 1, 1)^N$.

Let $(v, w)$ be a pair given by Theorem \ref{thm:mult-stab}, $j\in \{1,...,N\}$ and $y_0 \in \R$. Denote
\begin{equation}
\label{def:boIj}
\boI_{j,y_0}(t):=\int_\R \Phi(x-(a_j(t)+y_0)) [vw](x,t) \, dx.
\end{equation}

We prove a monotonicity formula for these localized versions of the momentum following the ideas used by Martel, Merle and Tsai in the proof of Lemma 3 in \cite{MarMeTs1}.
 
\begin{prop}
\label{prop:evol-momentum}
Let $y_0\in \R$, $t \in \R_+$ and $\sigma \in [-\delta_\gc,\delta_\gc]$. There exist positive numbers $\alpha_1 \leq \alpha$, $L_1 \geq L^*$ and $A_1,A^*_1>0$, depending only on $\gc$ and $\gs$, such that, if $\alpha_0 \leq \alpha_1$ and $L \geq L_1$, then the map $\boI_j$ is of class $\boC^1$ on $\R$, and it satisfies
\begin{equation}
\label{eq:deriv-Ij}
\begin{split}
\frac{d}{dt}\big[\boI_{j,y_0+\sigma t}(t)\big] \geq & \frac{\nu_\gc^2}{32} \int_\R \big[ (\partial_x v )^2 + v^2 + w^2 \big](x , t) \Phi'(x-(a_j(t)+y_0+\sigma t)) \, dx\\
&  - A_1 \exp \Big( - \frac{\nu_{\gc}}{16} |y_0 + \sigma t | \Big),
\end{split}
\end{equation}
for any $1 \leq j \leq N$ and any $t \in \R_+$. In particular, we have
\begin{equation}
\label{eq:monobisI}	
\boI_{j,y_0}(t_1) \geq \boI_{j,y_0}(t_0) - A^*_1 \exp \Big( - \frac{\nu_{\gc} }{16}|y_0| \Big),
\end{equation}
for any real numbers $t_1 \geq t_0\geq 0$.
\end{prop}

\begin{remark}
\label{rmk}
In view of the proof below, Proposition \ref{prop:evol-momentum} holds for any time $t \in \R$, when there is only one soliton in the sum. In particular, this further property is true for the limit solution $(\tilde{v}_j,\tilde{w}_j)$.
\end{remark}

\begin{proof}
We differentiate the quantities $\boI_{j,y_0+\sigma t}$ with respect to $t$ in order to obtain
\begin{equation}
\label{cronos}
\begin{split}
\frac{d}{dt}\big[\boI_{j,y_0+\sigma t}(t)\big] = & \frac{1}{2} \int_\R \Phi'(\cdot-(a_j(t)+y_0+\sigma t)) \times \\
& \ \ \ \ \ \ \ \times  \Big( v^2 + w^2 - \big( a_j'(t) + \sigma \big) v w - 3 v^2 w^2 + \frac{3 - v^2}{(1 - v^2)^2} (\partial_x v)^2 \Big)\\
& + \frac{1}{2} \int_\R \Phi'''(\cdot-(a_j(t)+y_0+\sigma t)) \, \ln \big( 1 - v^2 \big),
\end{split}
\end{equation}
for any $t \in \R_+$.
We decompose the real line into two regions, 
$$R_j(t) = \Big[ a_j(t) - \frac{L-1}{4}  , a_j(t) + \frac{L-1}{4} \Big],$$
and its complementary set. We set
$$\frac{d}{dt}\big[\boI_{j,y_0+\sigma t}(t)\big] = \boI^1_j(t) + \boI^2_j(t),$$
where
\begin{align*}
\boI^2_j(t) = & \frac{1}{2} \int_{R_j(t)} \Phi'(\cdot-(a_j(t)+y_0+\sigma t)) \times \\
& \ \  \ \ \ \ \times  \Big( v^2 + w^2 - \big( a_j'(t) + \sigma \big) v w - 3 v^2 w^2 + \frac{3 - v^2}{(1 - v^2)^2} (\partial_x v)^2 \Big)\\
& + \frac{1}{2} \int_{R_j(t)} \Phi'''(\cdot-(a_j(t)+y_0+\sigma t)) \, \ln \big( 1 - v^2 \big).
\end{align*}
When $x \in R_j(t)$, we have
$$|x-a_j(t)-y_0-\sigma t |\geq -\frac{L}{4} + |y_0 + \sigma t|.$$
Hence, using \eqref{eq:min-v}, \eqref{eq:lvmh2bis}, and \eqref{rhea2}, we obtain 

\begin{equation}
\label{ouranos}
\big| \boI^2_j(t) \big| \leq A_\gc \exp \Big( -\frac{\nu_{\gc} }{16}|y_0 + \sigma t | \Big),
\end{equation}
where $A_\gc$ denotes, here as in the sequel, a positive number depending only on $\gc$ and $\gs$.

Next, we use \eqref{eq:min-v} and \eqref{rhea2} to bound $\boI^1_j(t)$ from below by
\begin{equation}
\label{hyperion}
\begin{split}
\boI^1_j(t) \geq & \frac{1}{2} \int_{\R\setminus R_j(t)} \Phi'(\cdot-(a_j(t)+y_0+\sigma t)) \times \\
& \ \ \ \ \ \ \ \times \Big( (\partial_x v)^2 + v^2 + w^2 - 2 \Big( 1 - \frac{\nu_{\gc}^2}{4} \Big)^\frac{1}{2} |v| |w| - 3 v^2 w^2 + \frac{\nu_{\gc}^2}{64} \ln \big( 1 - v^2 \big) \Big).
\end{split}
\end{equation}
For any $x \in \R\setminus R_j(t)$, we have 
$$\big| x - a_k(t) \big| \geq \frac{L}{4},$$
for any $1 \leq k \leq N$. 
This yields, by \eqref{def:beps-t}, \eqref{eq:est-beps-t}, the Sobolev embedding theorem, the exponential decay of the solitons and \eqref{eq:est-coef-t} , that
$$\big| v(x, t) \big| \leq \big| \eps_v(x, t) \big| + \sum_{k = 1}^N \big| v_{c_k(t)}(x - a_k(t)) \big| \leq A_\gc \Big( \alpha + \exp \big( - \frac{\nu_{\gc}}{16} L \big) \Big),$$
for any $x \in \R \setminus R_j(t)$. For $\alpha$ small enough and $L$ big enough, we have
\begin{equation}
\label{ocean}
v^2 \leq \min \Big\{ \frac{1}{2}, \frac{\nu_{\gc}^2}{96} \Big\},
\end{equation}
on $ \R \setminus R_j(t)$. We conclude from \eqref{hyperion}, \eqref{ocean} and the fact that $\ln(1 - s) \geq - 2 s$ for $0 \leq s \leq 1/2$, that
$$\boI_j^1(t) \geq \frac{1}{2} \Big( 1 - \Big( 1 - \frac{\nu_{\gc}^2}{4} \Big)^\frac{1}{2} - \frac{\nu_{\gc}^2}{32} \Big) \int_{\R \setminus R_j(t)} \Phi'(\cdot-(a_j(t)+y_0+\sigma t)) \, \big( (\partial_x v)^2 + v^2 + w^2 \big).$$
Then, using the fact that $1 - (1 - s)^{1/2} \geq s/2$ for $0 \leq s \leq 1$, we obtain
\begin{align*}
\boI_j^1(t) &\geq \frac{1}{2} \Big( 1 - \Big( 1 - \frac{\nu_{\gc}^2}{4} \Big)^\frac{1}{2} - \frac{\nu_{\gc}^2}{32} \Big) \int_{\R \setminus R_j(t)} \Phi'(\cdot-(a_j(t)+y_0+\sigma t)) \, \big( (\partial_x v)^2 + v^2 + w^2 \big)\\
& \geq \frac{3 \nu_{\gc}^2}{64} \int_{\R \setminus R_j(t)} \Phi'(\cdot-(a_j(t)+y_0+\sigma t)) \, \big( (\partial_x v)^2 + v^2 + w^2 \big).
\end{align*}
 
This concludes the proof of \eqref{eq:deriv-Ij}. Now let us prove \eqref{eq:monobisI}. When $y_0 \geq 0,$ we integrate \eqref{eq:deriv-Ij} from $t_0$ to $\frac{t_1+t_0}{2}$ taking $\sigma =\frac{\delta_\gc}{2}$ and $y_0= y_0- \frac{\delta_\gc}{2} t_0$ and from $\frac{t_1+t_0}{2}$ to $t_1$ taking $\sigma =-\frac{\delta_\gc}{2}$ and $y_0= y_0+ \frac{\delta_\gc}{2} t_1$, to obtain \eqref{eq:monobisI}. The proof is similar when $y_0<0$. This finishes the proof of this proposition.
\end{proof}

Using Propositions \ref{prop:reprod} and \ref{prop:evol-momentum} and Remark \ref{rmk}, we claim as in \cite{Bahri} that
\begin{prop}[\cite{Bahri}]
\label{prop:local}
Let $t\in \R$. There exists a positive constant $\boA_{\gc^0}$ such that
$$\int_t^{t + 1} \int_\R \big[ (\partial_x \tilde{v}_j)^2 + \tilde{v}_j^2 + \tilde{w}_{j}^2 \big](x + \tilde{a}_j(s), s) e^{\frac{\nu_{\gc} }{16} |x|} \, dx \, ds \leq \boA_{\gc^0}.$$
\end{prop}
The two lemmas below are the main ingredients for the proof of this proposition. For the limit profile $(\tilde{v}_j, \tilde{w}_{j})$, we set $\tilde{\boI}_{j,\pm y_0}(t) := \boI_{j,\pm y_0}^{(\tilde{v}_j, \tilde{w}_{j})}(t)$ for any $t \in \R$ and any $y_0>0$.
\begin{lem}[\cite{BetGrSm2}]
\label{lem:loc}
For any positive number $\delta$, there exists a positive number $y_\delta$, depending only on $\delta$, such that for any $t \in \R$ we have
\begin{equation}
\label{eq:local-boI_tilde}
\Big|\tilde{\boI}_{j,y_0}(t)\Big| \leq \delta \quad {\rm and} \quad |P(\tilde{v}_j, \tilde{w}_j)-\tilde{\boI}_{j,-y_0}(t) | \leq \delta,
\end{equation}
for any $y_0 \geq y_\delta.$
\end{lem}
This lemma shows that the momentum of the limit profile is localized in a compact region of the real line. This is a key point to claim that this momentum is exponentially decaying with respect to $y_0$. 
\begin{proof}
The proof of this lemma is by contradiction. We assume that there exists a positive number $\delta_0$ such that, for any positive number $y_{0}$, there exists a number $ t_0 \in \R$ such that either $|\tilde{\boI}_{j,y_0}(t_0)| \geq \delta_0$ or $|\tilde{\boI}_{j,-y_0}(t_0) - P(\tilde{v}_j, \tilde{w}_j)| \geq \delta_0$. 

At initial time $t = 0$, we have $\lim_{y_0 \to +\infty} \tilde{\boI}_{j,y_0}(0) = \lim_{y_0 \to + \infty} \tilde{\boI}_{j,-y_0}(0) - P(\tilde{v}_j, \tilde{w}_j) = 0$. Hence, there exists $y_{0} > 0$ such that
\begin{equation}
\label{eq:sam0}
|\tilde{\boI}_{j,y_0}(0)| + |\tilde{\boI}_{j,-y_0}(0) - P(\tilde{v}_j, \tilde{w}_j)| \leq \frac{\delta_0}{4} \quad {\rm and} \quad A_\gc \exp \Big( - \frac{\nu_{\gc} }{16} y_0 \Big) \leq \frac{\delta_0}{32}.
\end{equation}
Now, we prove that the case  $\tilde{\boI}_{j,y_0}(t_0) \geq \delta_0$ cannot hold for this choice of $y_0$. The proof of the other cases can be written in a very similar manner. 

First, we deduce from \eqref{eq:sam0} that
$$\tilde{\boI}_{j,y_0}(t_0) \geq \delta_0 \geq \frac{\delta_0}{4} + \frac{\delta_0}{16} \geq \tilde{\boI}_{j,y_0}(0) + A_\gc \exp \Big( - \frac{\nu_{\gc} }{16} y_0 \Big) ,$$
Using \eqref{eq:monobisI}, we conclude that $t_0 > 0$. Next, from the fact that $ \lim_{y_0 \to + \infty} \tilde{\boI}_{j,-y_0}(t_0) - P(\tilde{v}_j, \tilde{w}_j) = 0$ we can choose $y'_0\geq y_0$ such that
\begin{equation}
\label{eq:sam1}
\big| \tilde{\boI}_{j,-y'_0}(t_0) - P(\tilde{v}_j, \tilde{w}_j) \big| \leq \frac{\delta_0}{4}.
\end{equation}
The choice of $y'_0$ can be done to conserve \eqref{eq:sam0} and to obtain
$$\big| \tilde{\boI}_{j,-y'_0}(t_0) - \tilde{\boI}_{j,y_0}(t_0) - P(\tilde{v}_j, \tilde{w}_j) \big| \geq \frac{3 \delta_0}{4} \quad {\rm and} \quad \big| \tilde{\boI}_{j,-y'_0}(0) - \tilde{\boI}_{j,y_0}(0) - P(\tilde{v}_j, \tilde{w}_j) \big| \leq \frac{\delta_0}{2},$$
and therefore
$$\Big| \big( \tilde{\boI}_{j,-y'_0}(0) - \tilde{\boI}_{j,y_0}(0) \big) - \big( \tilde{\boI}_{j,-y'_0}(t_0) - \tilde{\boI}_{j,y_0}(t_0) \big) \Big| \geq \frac{\delta_0}{4}.$$
Using the fact that the integrands of the expressions between parenthesis are compactly supported in the space, we infer from Proposition \ref{prop:reprod} that there exists an integer $n_0$ such that
$$\Big| \big( \boI_{j,-y'_0}(t_n) - \boI_{j,y_0}(t_n) \big) - \big( \boI_{j,-y'_0}(t_n + t_0) - \boI_{j,y_0}(t_n + t_0) \big) \Big| \geq \frac{\delta_0}{8},$$
for any $n \geq n_0$. Ordering well the terms in the previous inequality, we obtain
\begin{equation}
\label{eq:sam3}
\max \Big\{ \big| \boI_{j,-y'_0}(t_n) - \boI_{j,-y'_0}(t_n + t_0) \big|, \big| \boI_{j,y_0}(t_n) - \boI_{j,y_0}(t_n + t_0) \big| \Big\} \geq \frac{\delta_0}{16}.
\end{equation}
Since $t_0 \geq 0$, by \eqref{eq:monobisI}, and \eqref{eq:sam0}, we deduce 
$$\boI_{j,-y'_0}(t_n) - \boI_{j,-y'_0}(t_n + t_0) \leq \frac{\delta_0}{32} \quad {\rm and} \quad \boI_{j,y_0}(t_n) - \boI_{j,y_0}(t_n + t_0) \leq \frac{\delta_0}{32},$$
and then we infer from \eqref{eq:sam3} that, for any $n \geq n_0$,
$${\rm either} \quad \boI_{j,-y'_0}(t_n + t_0) - \boI_{j,-y'_0}(t_n) \geq \frac{\delta_0}{16}, \quad {\rm or} \quad \boI_{j,y_0}(t_n + t_0) - \boI_{j,y_0}(t_n) \geq \frac{\delta_0}{16}.$$
This leads us to the possibility of choosing an increasing sequence $(n_k)_{k \in \N}$ such that \\ $t_{n_{k + 1}} \geq t_{n_k} + t_0$ for any $k \in \N$, and either
\begin{equation}
\label{eq:sami0}
\boI_{j,y_0}(t_{n_k} + t_0) - \boI_{j,y_0}(t_{n_k}) \geq \frac{\delta_0}{16},
\end{equation}
for any $k \in \N$, or
$$\boI_{j,-y'_0}(t_{n_k} + t_0) - \boI_{j,-y'_0}(t_{n_k}) \geq \frac{\delta_0}{16},$$
for any $k \in \N$. Next, we suppose that \eqref{eq:sami0} holds, the proof of the other case being exactly the same. From the fact that $t_{n_{k + 1}} \geq t_{n_k} + t_0$, we conclude using \eqref{eq:monobisI}, \eqref{eq:sam0} and \eqref{eq:sami0}, that
\begin{equation}
\label{eq:sami1}
\boI_{j,y_0}(t_{n_{k + 1}}) \geq \boI_{j,y_0}(t_{n_k} + t_0) - \frac{\delta_0}{32} \geq \boI_{j,y_0}(t_{n_k}) + \frac{\delta_0}{32},
\end{equation}
for any $k \in \N$. Now, we recall that 
$\boI_{j,y_0}(t_{n_k})$ is bounded by the energy of the initial datum. This yields a contradiction with \eqref{eq:sami1} and finishes the proof. 
\end{proof}
At this stage, the problem reduces to the case of one soliton. The proof of the next statement is exactly the same as the one given by the author in \cite{Bahri} for that case (see also \cite{BetGrSm2} for more details).
\begin{lem}[\cite{Bahri}]
Let $y_0>0$. For any $t \in \R$ we have
\begin{equation}
\label{eq:local-boI_tilde1}
\tilde{\boI}_{j,y_0}(t) \leq A_\gc \exp \Big( - \frac{\nu_{\gc} }{16} y_0 \Big) \quad {\rm and} \quad |P(\tilde{v}_j, \tilde{w}_j)-\tilde{\boI}_{j,-y_0}(t) | \leq A_\gc \exp \Big( - \frac{\nu_{\gc} }{16} y_0 \Big).
\end{equation}
\end{lem}

The proof of Proposition \ref{prop:smooth} is then exactly the same as the one of Proposition 2.7 in \cite{Bahri}.

\section{Asymptotic stability between the solitons in the hydrodynamical framework}

\subsection{Proof of \eqref{conv-stab-asymp1}}

Let $\gc^0$ be as in Theorem \ref{thm:stabasympt} and $\gv_0$ be any pair which belongs to the set $\boV(\alpha,L)$ with $\alpha$ and $L$ as in the hypothesis of Theorem \ref{thm:stabasympt}.

Let $j \in \{1,\ldots,N\}$ and $b_j$ satisfying \eqref{cond1:b_j}--\eqref{cond3:b_j}. By \eqref{eq:est-beps-t}, $\varepsilon$ is uniformly bounded in $X(\R)$. Then, there exists $\varepsilon^*_{j,0} \in X(\R)$ such that, up to a subsequence,
\begin{equation}
\label{conv:eps1}
\varepsilon(t_n,.+b_j(t_n)) \rightharpoonup \varepsilon^*_{j,0} \quad \hbox{in} \ \ X(\R) \quad \hbox{as} \quad n\to +\infty.
\end{equation}

We set $\gv^*_{j,0}=(v^*_{j,0},w^*_{j,0}):= \varepsilon^*_{j,0}$ and denote by $\gv^*_j=(v^*_j,w^*_j)$ the unique global solution to \eqref{HLL} corresponding to this initial datum $\gv^*_{j,0}$. We claim that this solution exponentially decays with respect to the space variable for any time, as well as all its space derivatives. More precisely, we have
\begin{prop}
\label{prop:smooth1}
The pair $(v^*_j,w^*_j)$ is indefinitely smooth and exponentially decaying on $\R \times\R$. Moreover, given any $k \in \N$, there exists a positive constant $A_{k, \gc}$, depending only on $k$ and $\gc$, such that
\begin{equation}
\label{eq:smooth1}
\int_\R \big[ (\partial_x^{k + 1} v^*_j)^2 + (\partial^k_x v^*_j)^2 + (\partial_x^k w^*_j)^2 \big](x + \tilde{b}_j(t), t) \exp\big( \frac{\nu_\gc}{16} |x|\big) \, dx \leq A_{k, \gc},
\end{equation}
for any $t\in \R$, where $\tilde{b}_j$ satisfies \eqref{cond1:b_j}--\eqref{cond3:b_j}.
\end{prop}
In view of this proposition, we can establish a Liouville type theorem in order to finish the proof of Theorem \ref{thm:stabasympt}. 
\begin{prop}
\label{th:liouville1}
There exists a positive number $\alpha^*$ such that, if $(v,w)$ is a solution of \eqref{HLL} satisfying \eqref{eq:smooth1} and 
$$\|(v_{0},w_{0})\|_{X(\R)} \leq \alpha^* ,$$ 
then, 
$$(v,w)(t,x)=0 \quad \forall (t,x) \in \R\times \R .$$
\end{prop}

This result concludes the proof of Theorem \ref{thm:stabasympt} since $\varepsilon^*_{j,0} \equiv 0$ for any sequence $(t_n)_{n\in \N}$. Indeed, if we suppose that there exists a sequence of time $(s_n)$ such that $\varepsilon\*_{j,0}\neq 0,$ then, in view of the previous analysis, we get a contradiction from Proposition \ref{th:liouville1}.

Now, we will show that \eqref{conv:asymstab1} holds also when $b_j$ is an arbitrary map satisfying \eqref{cond1:b_j} and \eqref{cond3inf:b_j} instead of \eqref{cond3:b_j}.
\begin{proof}
Let $(t_n)$ be a sequence of time such that $t_n \to \infty$ as $n\to \infty$. It follows from \eqref{cond3inf:b_j}, up to a subsequence, $\frac{b_j(t_n)}{t_n}$ has a limit $l_j$ as $n\to \infty$ and $c_{j-1}^\infty < l_j < c_{j}^\infty.$
Next, we take $\tilde{b}_j$ a smooth extension of $b_j$ such that $\tilde{b}_j(t_n)=b_j(t_n)$ for all $n \in \N.$ More precisely, $\tilde{b}_j \in \boC^1(\R_+,\R)$ verifies \eqref{cond1:b_j},  and, from \eqref{cond3inf:b_j}, we have
$$\lim_{t\to \infty} \tilde{b}'_j(t)= \lim_{n\to \infty} \frac{\tilde{b}_j(t_n)}{t_n}= l_j.$$
Hence, $\tilde{b}_j$ satisfies \eqref{cond3:b_j}. Then, by \eqref{conv-stab-asymp1}, we obtain
$$(v,w)(t_n,\cdot+\tilde{b} _j(t_n)) \rightharpoonup 0 \quad {\rm in} \ X(\R),$$
as $n \to \infty.$ This leads to
$$(v,w)(t_n,\cdot+b _j(t_n)) \rightharpoonup 0 \quad {\rm in} \ X(\R),$$
as $n \to \infty.$ This finishes the proof since this convergence holds for any sequence $(t_n)$ such that $t_n \to +\infty$ as $n\to +\infty$.
\end{proof}
In the next two subsections we begin by proving Proposition \ref{th:liouville1} and then we give the proof of Proposition \ref{prop:smooth1}.
\subsection{Proof of the Liouville type theorem}
First, we verify that our limit solution has a small norm. This is a direct consequence of the conservation of the energy, \eqref{conv:eps1}, Theorem \ref{thm:mult-stab} and equivalence between the energy and the norm of $X(\R)$. More precisely, we have
$$ \|(v^*_{j,0},w^*_{j,0})\|_{X(\R)} \leq  \liminf_{n\to \infty} \|\varepsilon(t_n)\|_{X(\R)} \leq A_\gc \alpha,$$
and then, 
$$ \|(v^*_{j},w^*_{j})(t)\|_{X(\R)} \leq A_\gc \boE\big(v^*_{j},w^*_{j}\big)(t)=A_\gc \boE\big(v^*_{j,0},w^*_{j,0}\big)  \leq A_\gc \|(v^*_{j,0},w^*_{j,0})\|_{X(\R)} \leq A_\gc \alpha,$$
for all $t\in ]T_-,T_+[$, where $]T_-,T_+[$ denotes the maximal interval of existence for the solution $(v^*_{j},w^*_{j})$. We derive from this inequality the existence of a number $0<\delta<1$ such that 
$$\|v^*_j(t)\|_{L^\infty}\leq \delta <1,$$
for all $t\in ]T_-,T_+[.$ It then follows from the result in \cite{DeLGr} that the solution $(v^*_{j},w^*_{j})$ is actually global, and that it satisfies
\begin{equation}
\label{norm:small}
\|(v^*_{j},w^*_{j})(t)\|_{X(\R)} \leq A_\gc \boE\big(v^*_{j},w^*_{j}\big)(t) \leq A_\gc \alpha.
\end{equation}
for all $t \in \R$.

Next, we linearise \eqref{HLL} around zero. Let $\gv:=(v,w)$ be a solution of \eqref{HLL} verifying \eqref{norm:small}. We obtain
\begin{equation}
\label{eq:linearized}
\partial_t \gv = J L \gv + J B \gv ,
\end{equation}
where we have denoted
\begin{equation}
\label{def:J}
J =   S \partial_x := \begin{pmatrix} 0 &   \partial_x \\  \partial_x & 0 \end{pmatrix},
\end{equation}
$$ L \gv := \begin{pmatrix} - v + \partial_{xx} v \\ - w \end{pmatrix},$$
and
$$ B \gv := \begin{pmatrix} \frac{( \partial_{xx} v ) v^2}{1 - v^2} + \frac{(\partial_x v)^2v}{(1 - v^2)^2} + v w^2 \\ v^2 w \end{pmatrix}.$$
Now, we consider the following quantity
$$U(t):=\int_\R x [v_j^*w_j^*](t,x) dx,$$
for any $t \in \R$. Since $(v^*_j,w^*_j)$ is a smooth solution of \eqref{HLL} which satisfies \eqref{eq:smooth1}, the map $U$ is of class $\boC^1$ and it is possible to differentiate the integrand with respect to the time variable. Hence, we deduce from \eqref{eq:linearized} and an integration by parts that
\begin{equation}
\label{deriv:U}
U'(t) = - \big< L \gv_j^*(t), \gv_j^*(t) \big> _{L^2(\R)} - \big< L \gv_j^*(t), \mu \partial_x \gv_j^*(t) \big> _{L^2(\R)} + \big< \mu \partial_x B \gv_j^*,  \gv_j^* \big> _{L^2(\R)},
\end{equation}
where $\mu(x)=x$ for all $x \in \R$.
For the linear terms, we integrate by parts to write
\begin{equation}
\label{part:lin}
-\big< L \gv_j^*(t), \gv_j^*(t) \big> _{L^2(\R)} - \big< L \gv_j^*(t), \mu \partial_x \gv_j^*(t) \big> = \int_\R \big[\frac{3}{2} (\partial_x v_j^*(t))^2 + \frac{1}{2} (v_j^*(t))^2 + \frac{1}{2} (w_j^*(t))^2 \big].
\end{equation}
For the other term, we use the Cauchy-Schwarz inequality, the Sobolev embedding theorem, \eqref{eq:smooth1} and \eqref{norm:small} to infer that
\begin{equation}
\label{part:nonlin}
\big| \big< \mu \partial_x B \gv_j^*,  \gv_j^* \big> _{L^2(\R)} \big| \leq A_\gc \alpha \|\gv_j^*\|^2_{X(\R)}.
\end{equation}
Indeed, let us estimate two terms of the right hand side. The other ones can be estimated in a very similar way. Performing integrations by parts, and using the Cauchy-Schwarz inequality and \eqref{eq:max-v}, we can write
\begin{align*}
\Big| \int_\R x \partial_x \big((v^*_j(t,x))^2 w^*_j(t,x) \big) w^*_j(t,x) dx \Big| \leq & \|\mu \partial_x w^*_j(t)\|_{L^\infty} \| v^*_j(t)\|_{L^\infty} \| v^*_j(t)\|_{L^2} \| w^*_j(t)\|_{L^2}\\
& +  \| v^*_j(t)\|^2_{L^\infty} \| w^*_j(t)\|^2_{L^2}, 
\end{align*}
and 
\begin{align*}
\Big| \int_\R x \partial_x \Big(\frac{( \partial_{xx} v^*_j(t,x) ) (v^*_j)^2(t,x)}{1 - (v^*_j)^2(t,x)} \Big) v^*_j(t,x) dx \Big| \leq & A_\gc \|\mu \partial_{xx} v^*_j(t)\|_{L^\infty} \|\partial_{x} v^*_j(t)\|_{L^2} \| v^*_j(t)\|_{L^2}\| v^*_j(t)\|_{L^\infty} \\
& + A_\gc \| \partial_{xx} v^*_j(t)\|_{L^\infty} \| v^*_j(t)\|_{L^\infty} \| v^*_j(t)\|^2_{L^2}.
\end{align*}
Then, by the Sobolev embedding theorem, \ref{eq:smooth1} and \eqref{norm:small}, we obtain
$$\Big| \int_\R x \partial_x \big((v^*_j(t,x))^2 w^*_j(t,x) \big) w^*_j(t,x) dx \Big| \leq A_\gc \alpha \|\gv_j^*(t)\|^2_{X(\R)},$$
and
$$\Big| \int_\R x \partial_x \Big(\frac{( \partial_{xx} v^*_j(t,x) ) (v^*_j)^2(t,x)}{1 - (v^*_j)^2(t,x)} \Big) v^*_j(t,x) dx \Big| \leq A_\gc \alpha \|\gv_j^*(t)\|^2_{X(\R)}.$$

Now, we introduce \eqref{part:lin} and \eqref{part:nonlin} into \eqref{deriv:U} and we choose $\alpha$ small enough to claim that
\begin{equation}
\label{bnd:norm}
U'(t) \geq \frac{1}{4} \|\gv_j^*(t)\|^2_{X(\R)}.
\end{equation}

Since $U$ is uniformly bounded on $\R$, we infer that the map $t\mapsto \|\gv_j^*(t)\|_{X(\R)}$ belongs to $L^2(\R).$ This yields the existence of a sequence of positive times $(s_n)_{n \in \N}$, which goes to $+\infty$ as $n \to +\infty$, such that we have
\begin{equation}
\label{limit:norm}
\lim_{n \to \infty} \|\gv_j^*(\pm s_n)\|_{X(\R)}=0.
\end{equation}
In view of \eqref{eq:smooth1}, this gives 
$$\lim_{n \to \infty} U(\pm s_n) = 0 .$$
Integrating \eqref{bnd:norm} from $-s_n$ to $s_n$ and taking the limit $n \to +\infty$, we deduce that
$$\int_\R \|\gv_j^*(t)\|^2_{X(\R)} dt =0.$$
Hence, $$\gv_j^*\equiv 0 \quad {\rm on} \quad \R \times \R .$$
This finishes the proof of Theorem \ref{th:liouville1}.  \qed  
\subsection{Proof of Proposition \ref{prop:smooth1}}
In this section, we prove the exponential decay of the limit solution $\gv^*_j$. First, we state the monotonicity of the momentum.
Let $(v, w)$ be a pair given by Theorem \ref{thm:mult-stab}, $j\in \{1,...,N\}$ and $y_0 \in \R$. Denote
$$\boI_{j,y_0}(t):=\int_\R \Phi(x-(b_j(t)+y_0)) [vw](x,t) \, dx,$$
for $b_j$ satisfying \eqref{cond1:b_j} and \eqref{cond3:b_j} and set
$$\lambda_{\gc,\gamma} := \frac{1}{2} \min \big\{ 1+\gamma_1, \gamma_2-c_1, c_2-\gamma_2, \ldots ,\gamma_{N+1}-c_{N}, 1-\gamma_{N+1}\big\}$$
for any $\gc \in (- 1, 1)^N$, where $\gamma:= (\gamma_1,\ldots,\gamma_{N+1}):=\lim_{t\to \infty}\big(b_1'(t),\ldots,b'_{N+1}(t)\big)$. We claim the following monotonicity formula for this localized version of the momentum.
 
\begin{prop}
\label{prop:evol-momentum1}
There exist positive numbers $\alpha_2 \leq \alpha$, $L_2 \geq L^*$, $T>0$ and $A_2,A^*_2>0$, depending only on $\gc$ and $\gs$, such that, if $\alpha_0 \leq \alpha_2$ and $L \geq L_2$, then the map $\boI_{j,y_0}$ is of class $\boC^1$ on $\R$, and it satisfies
\begin{equation}
\label{eq:deriv-Ij1}
\begin{split}
\frac{d}{dt}\big[\boI_{j,y_0}(t)\big] \geq & \frac{\nu_\gc^2}{32} \int_\R \big[ (\partial_x v )^2 + v^2 + w^2 \big](x , t) \Phi'(x-(b_j(t)+y_0)) \, dx\\
&  - A_2 \exp \Big( - \frac{\nu_{\gc}}{16} ( |y_0 + \lambda_{\gc,\gamma} t | \Big),
\end{split}
\end{equation}
for any $1 \leq j \leq N$ and any $t \geq T$. In particular, we have
\begin{equation}
\label{eq:monobisI1}	
\boI_{j,y_0}(t_1) \geq \boI_{j,y_0}(t_0) - A^*_2 \exp \Big( - \frac{\nu_{\gc} }{16}|y_0| \Big),
\end{equation}
for any real numbers $t_1 \geq t_0 \geq T$.
\end{prop}
The proof is very similar to the one of Proposition 5 in \cite{DeLGr}. We will only sketch it.
\begin{proof}
As in the proof of Proposition \ref{prop:evol-momentum}, we write
$$\boI_{j,y_0}'(t) = \boI_1(t) + \boI_2(t),$$
decomposing the real line into the region $I_j(t)$ and its complementary set, where $I_j(t)$ is the interval defined by
$$I_j(t) = \Big[ b_j(t) - \frac{1}{4} \big( L + \lambda_{\gc,\gamma} t \big), b_j(t) + \frac{1}{4} \big( L + \lambda_{\gc,\gamma} t \big) \Big].$$ 
For $\boI_2$, we have (see the proof of Proposition \ref{prop:evol-momentum} for more details)
$$\big| \boI_2(t) \big| \leq A^* \exp \Big( - \frac{1}{32} \big( L + \lambda_{\gc,\gamma} t \big) \Big).$$
For $\boI_1(t)$, we first infer from \eqref{cond1:b_j} that there exists $T>0$ sufficiently large such that for all $t \geq T$,
$$ c_{j-1}^\infty < b'_j(t) < c_{j}^\infty,$$
and then
$$ b'_j(t)^2 \leq 1 - \frac{\nu_{\gc}^2}{4}.$$
This leads, using \eqref{eq:min-v} and \eqref{rhea2}, to
$$\boI_1(t) \geq \frac{1}{2} \int_{I_j(t)} \Phi'\big( \cdot - (b_j(t)+y_0)\big) \, \Big( (\partial_x v)^2 + v^2 + w^2 -  2 \Big( 1 - \frac{\nu_{\gc}^2}{4} \Big)^\frac{1}{2} |v| |w| - 3 v^2 w^2 + \frac{\nu_{\gc^*}^2}{64} \ln \big( 1 - v^2 \big) \Big).$$
Now, increasing the value of $T>0$ if necessary,  we infer from \eqref{cond3:b_j} that
$$\big| a_k(t)- b_j(t)| \geq \frac{1}{2} \big( L  + \lambda_{\gc,\gamma} t \big),$$
for any $t\geq T$ and $1 \leq k \leq N$. 
When $x \in I_j(t)$, we have
$$\big| x - a_k(t) \big| \geq \Big| a_k(t) - b_j(t) \Big| - \frac{1}{4} \big( L + \lambda_{\gc,\gamma} t \big) \geq \frac{1}{4} \big( L + \lambda_{\gc,\gamma} t \big),$$
for any $1 \leq k \leq N$. This yields, using \eqref{def:beps-t}, \eqref{eq:est-beps-t} (and the Sobolev embedding theorem), \eqref{eq:est-coef-t} and the exponential decay of the solitons,
$$\big| v(x, t) \big| \leq \big| \eps_1(x, t) \big| + \sum_{k = 1}^N \big| v_{c_k(t)}(x - a_k(t)) \big| \leq A^* \Big( \alpha + \exp \Big( - \frac{\nu_{\gc^*}}{16} \big( L + \lambda_{\gc,\gamma} t \big) \Big) \Big),$$
for any $x \in I_j(t)$. We now decrease $\alpha$ and increase $L$, if necessary, to guarantee that $|v|$ is sufficiently small on the interval $I_j(t)$. Then we can finish the proof as the one of Proposition \ref{prop:evol-momentum1}.
\end{proof}
\begin{remark}
\label{rmk1}
In view of the proof below, the limit solution $(v^*_j,w^*_j)$ satisfies the conclusions of Proposition \ref{prop:evol-momentum1} for any time $t \in \R$, .
\end{remark}
The following claim contains the weak continuity of the flow and the convergence of the parameter $b_j$.
\begin{claim}
\label{claim:reprod}
Let $j \in \{ 1, \ldots ,N \}$ and $t \in \R$ be fixed. Then, there exist a map $b^*_j \in C^1(\R,\R)$ verifying   \eqref{cond1:b_j}-\eqref{cond3:b_j} such that
\begin{equation}
\label{w-conv1}
\big( v(\cdot + b_j(t_n), t_n + t), w(\cdot + b_j(t_n), t_n + t) \big) \rightharpoonup \big( v^*_j(\cdot, t), w^*_{j}(\cdot, t) \big)
\end{equation}
and
\begin{equation}
\label{w-conv2}
\big( v(\cdot + b_j(t_n+t), t_n + t), w(\cdot + b_j(t_n+t), t_n + t) \big) \rightharpoonup \big( v^*_j(\cdot + b^*_j(t) , t), w^*_{j}(\cdot + b^*_j(t), t) \big)
\end{equation}
in $X(\R)$, while
\begin{equation}
\label{conv:b_j}
b_j(t_n + t) - b_j(t_n) \to b^*_j(t),
\end{equation}
as $n \to + \infty$.
\end{claim}
\begin{proof}
We take $b_j^*(t):=\gamma_j t,$ for all $t \in \R$, where $\gamma_j:= \lim_{t\to +\infty}b'_j(t)$. Clearly, $b^*_j$ satisfies \eqref{cond1:b_j}-\eqref{cond3:b_j}. Then, the proof remains exactly the same as the one of Proposition \ref{prop:reprod}.
\end{proof}
As in the previous section, we claim the following lemma which shows the localisation of the momentum for the limit solution. For the limit profile $(v^*_j, w^*_{j})$, we set 
$$\boI^*_{j,\pm y_0}(t) := \boI_{j,\pm y_0}^{(v^*_j, w^*_{j})}(t) = \int_\R [v^*_j w^*_{j}](t) \Phi(\cdot - ( \pm y_0 + b^*_j(t))),$$ 
for any $t \in \R$ and $y_0>0$.
\begin{lem}[\cite{BetGrSm2}]
For any positive number $\delta$, there exists a positive number $y_\delta$, depending only on $\delta$, such that for any $t \in \R$ we have
\begin{equation}
\label{eq:local-boI^*}
\big|\boI^*_{j,y_0}(t)\big| \leq \delta \quad {\rm and} \quad |P(v^*_j, w^*_j)-\boI^*_{j,-y_0}(t) | \leq \delta,
\end{equation}
for any $y_0 \geq y_\delta.$
\end{lem}
In view of Remark \ref{rmk1}, the proof is similar to the one of Lemma \ref{lem:loc}. 

We also have 
\begin{lem}[\cite{BetGrSm2}]
Let $y_0>0$. For any $t \in \R$, we have
\begin{equation}
\label{eq:local-boI^*1}
\boI^*_{j,y_0}(t) \leq A_\gc \exp \Big( - \frac{\nu_{\gc}}{16} y_0 \Big) \quad {\rm and} \quad |P(v^*_j,w^*_j)-\boI^*_{j,-y_0}(t) | \leq A_\gc \exp \Big( - \frac{\nu_{\gc} }{16} y_0 \Big).
\end{equation}
\end{lem}
Using Proposition \ref{prop:evol-momentum1}, we claim as in \cite{BetGrSm2} that
\begin{prop}[\cite{Bahri}]
\label{prop:local1}
Let $t\in \R$. There exists a positive constant $\boA_{\gc^0}$ such that
$$\int_t^{t + 1} \int_\R \big[ (\partial_x v^*_j)^2 + (v^*_j)^2 + (w^*_{j})^2 \big](x + b^*_j(s), s) e^{\frac{\nu_{\gc} }{16} |x|} \, dx \, ds \leq \boA_{\gc^0}.$$
\end{prop}
At this stage, the proof of Proposition \ref{prop:smooth1} remains exactly the same as in \cite{Bahri} (see Section 4.2 for more details).
\begin{merci}
I would like to thank both my supervisors R. Côte and P. Gravejat for their support throughout the time I spent to finish this manuscript, offering invaluable advices.
This work is supported by a PhD grant from "Région Ile-de-France" 
\end{merci}
\bibliographystyle{plain}

\end{document}